%% file: main_arvix.tex
\documentclass[letterpaper, 10 pt, conference]{ieeeconf}  

\IEEEoverridecommandlockouts                              
\usepackage[utf8]{inputenc}
\usepackage[T1]{fontenc}
\usepackage{graphicx}
\graphicspath{ {./images/} }

\usepackage{sgame}
\usepackage{amssymb}
\usepackage{amsmath}
\usepackage{xcolor}
\usepackage{float}
\usepackage{comment}
\usepackage{lipsum}
\usepackage{mathbbol}
\usepackage[capitalize]{cleveref}
\usepackage{enumerate}
\usepackage{cite}
\usepackage{mathtools}

\newtheorem{remark}{Remark}
\newtheorem{proposition}{Proposition}
\newtheorem{theorem}{Theorem}

\newtheorem{definition}{Definition}

\newtheorem{lemma}{Lemma}

\DeclareMathOperator{\Equaldef}{\overset{def}{=}}
\input{symbols}

\renewcommand{\indic}[1]{\mathbb{1}(#1)}

\title{\LARGE \bf Strategic Delay and Coordination Efficiency in Global Games}

\author{Shinkyu Park, Behrouz Touri, and Marcos M.\ Vasconcelos   
\thanks{The authors contributed equally to this work and are listed in alphabetical order by their last names.}
\thanks{S.\ Park is with 
Electrical and Computer Engineering, King Abdullah University of Science and Technology, Thuwal, 23955, Saudi Arabia, B. Touri is with the Department of Industrial and Systems Engineering, University of Illinois Urbana Champaign, Champaign, IL 61820, and M.\ M.\ Vasconcelos are with the Department of Electrical and Computer Engineering, FAMU-FSU College of Engineering, Florida State University,  Tallahassee, FL 32306, USA. E-mails:
        \{\tt  shinkyu.park@kaust.edu.sa, touri1@illinois.edu, m.vasconcelos@fsu.edu\}.}%
}

\begin{document}

\maketitle
\thispagestyle{empty}
\pagestyle{empty}

\begin{abstract}

We investigate a coordination model for a two-stage collective decision-making problem within the framework of global games. The agents observe noisy signals of a shared random variable, referred to as the \textit{fundamental}, which determines the underlying payoff. Based on these signals, the agents decide whether to participate in a collective action now or to delay. An agent who delays acquires additional information by observing the identities of agents who have chosen to participate in the first stage. This informational advantage, however, comes at the cost of a discounted payoff if coordination ultimately succeeds. Within this  decision-making framework, we analyze how the option to delay can enhance collective outcomes. We show that this intertemporal trade-off between information acquisition and payoff reduction can improve  coordination and increase the efficiency of collective decision-making.

\end{abstract}


\section{Introduction}

Consider a collection of agents sharing a common environment. A task with a random difficulty is introduced, and the agents must coordinate to complete it. However, the task's difficulty (defined as the number of agents required to complete the task and obtain a positive reward) is imperfectly observed through noisy channels. The agents face a two-stage decision process regarding whether to participate on the task or not. In the first stage, agents may take a risky action and receive an undiscounted payoff. The agents who choose to delay their decision to the second stage receive an additional information signal, observing the set of agents who took the risky action in the first stage, but at the cost of a discounted payoff. From the perspective of a system designer, we seek to determine whether incorporating this second stage is advantageous in equilibrium or not. Addressing this requires evaluating a nontrivial trade-off between the initial noise level in the private signals about the task difficulty, the informational gain achieved by delaying the decision, and the penalty of the discounted payoff. In this paper, we formalize this trade-off to answer this question.

\begin{figure}[t!]
    \centering
\includegraphics[width=0.9\columnwidth]
{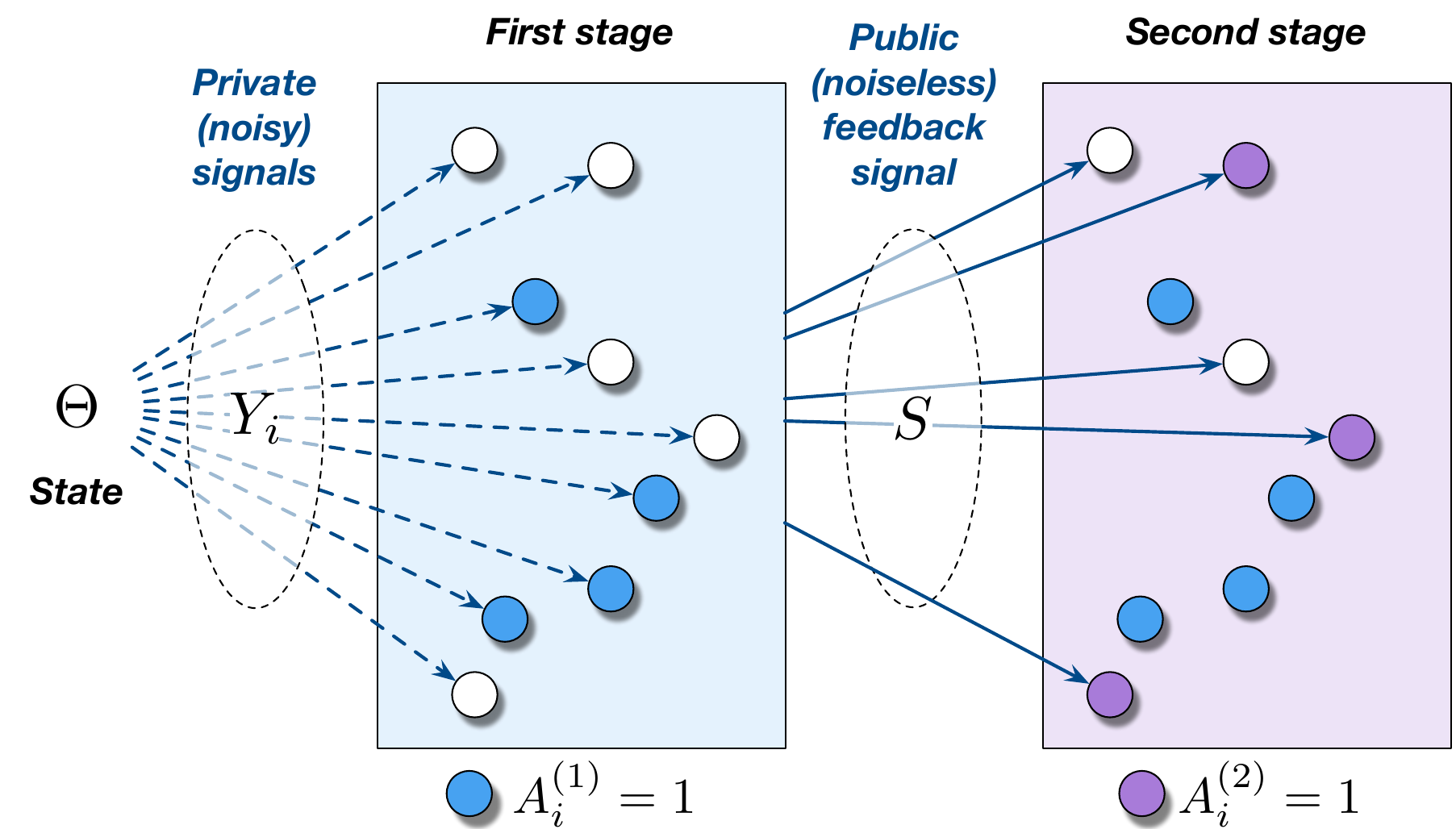}
    \caption{Two-stage global game with public (noiseless) feedback signal. In the first stage, the blue nodes decided to take the risky action, and the white nodes delayed their decision to the second stage. In the second stage the purple nodes decided to take the risky action. The payoff is collected at the end of the game.}
    \label{fig:two_stage}
\end{figure}

The applications of this decision-making framework are broad, spanning engineering \cite{Krishnamurthy:2009,Kanakia:2016,beaver2025multi,Wei:2023}, economics \cite{jin2023coordination}, biology \cite{vasconcelos2025multi}, and political science \cite{AngeletosHellwigPavan2007}. Specific examples include distributed task allocation in multi-robot systems, speculative currency attacks, bank runs, risky investments, the disclosure of private preferences, microbial infections, and political revolutions \cite{Mahdavifar:2017,dahleh2016coordination,Vasconcelos:2023}. Traditionally, these problems are modeled as \textit{global games}. This literature was initiated by the seminal work of \cite{Carlsson:1993} and further developed and popularized by \cite{Morris:2003}. 

The vast majority of this literature focuses on single-stage games. Due to the intrinsic complexity of even the most canonical problem formulations, the analysis of multi-stage global games has remained limited. The work of Dasgupta~\cite{Dasgupta2007CoordinationDelayGlobalGames} is among the first to study the option to delay in global games, characterizing coordination regimes in a continuum of agents with binary payoffs. Notably, \cite{Dasgupta2007CoordinationDelayGlobalGames} relies on the assumption of vanishing noise, which serves as an analytical device to ensure desirable properties such as equilibrium uniqueness. In contrast, our approach focuses on the non-vanishing noise regime, although the noise level may still be sufficiently small. In particular, we are interested in scenarios where the noise level could also be  manipulated. For example, when an incumbent regime manipulates information to induce instability and hinder coordination within a population. Another arises when a cyber-attacker seeks to disrupt the operations of a swarm robotic system cooperating to fulfill tasks in a smart warehouse.

Building on the setup introduced by \cite{Dasgupta2007CoordinationDelayGlobalGames}, we consider a payoff structure that is more closely aligned with the growing engineering and operations research literature on global games. Specifically, the payoff increases linearly with the number of agents taking the risky action and decreases linearly with the task difficulty (i.e., the fundamental). We assume Gaussian prior and noise distributions. In contrast to much of the economics literature, we do not impose degenerate priors or vanishing noise. Instead, we characterize equilibria for given model parameters, allowing for a more direct connection to practical settings. 

We begin with a finite population and show that the introduction of a delay option facilitates coordination from the perspective of a mechanism designer. We then analyze the infinite-population limit, where we first identify noise regimes under which the equilibrium in the class of threshold policies is unique, and subsequently derive comparative statics with respect to the discount factor. Finally, by introducing a notion of coordination efficiency, we demonstrate that there exist regimes in which adding a second stage improves equilibrium efficiency and others when it does not. 

The remainder of the paper is organized as follows. Section~II introduces the model, and Section~III studies equilibrium properties in the finite-agent regime. Section~IV analyzes the infinite-agent regime, where equilibrium reduces to a single threshold parameter and admits tractable analysis; comparative statics are also presented. Section~V introduces coordination efficiency, and Section~VI concludes the paper.

\section{System Model}

Our framework consists of a two-stage stochastic game with partial information. In this section, we introduce the key components of the game and the associated solution concept.

\subsection{Agents and actions}
Consider a game with a set of agents $\ags \Equaldef \{1, \ldots, N\}$. We study both the finite-population case ($N < \infty$) and the infinite-population limit ($N \to \infty$), in which case $\ags = \mathbb{N}$. The decision process unfolds over two stages, indexed by $t \in \{1,2\}$. At each stage, each agent $i \in \ags$ selects an action from the binary set $\mathcal{A}_i\Equaldef\{0,1\}$. Following the convention in the global games literature, action $1$ is referred to as the \emph{risky}, while action $0$ is the \emph{safe}. The action of the $i$-th agent at time $t$ is denoted by $A_i^{(t)}$.

The agents are allowed to take the risky action at most once across the two stages. This imposes the following constraint 
$A_i^{(1)} + A_i^{(2)} \leq 1, \, i \in \ags$ on the action sequence.
We let $F_i\Equaldef A_i^{(1)} + A_i^{(2)}$ to be the aggregate action of the player over the two stages.

\subsection{Information structure and feedback mechanism}
The environment is characterized by an underlying fundamental state $\Theta$, drawn from a standard normal distribution, i.e., $\Theta \sim \mathcal{N}(0,1)$. The agents do not observe $\Theta$ directly. Instead, prior to the first decision stage, each agent $i \in \ags$ receives a private signal $Y_i$, which is a noisy observation of the fundamental. More precisely, we assume 
\begin{align} \label{eq:noisy_observation}
Y_i = \Theta + Z_i, \quad i \in \ags,
\end{align}
where the idiosyncratic noise $Z_i$ is independently distributed across agents according to $Z_i \sim \mathcal{N}(0,\sigma^2)$.

To capture the dynamics of delayed decision-making, we introduce a \textit{public feedback mechanism}. At the conclusion of the first stage $t=1$,  all agents observe the proportion of early adopters who chose to take the risky action\footnote{In this work, our focus will be on symmetric policies and due to the anonymous utility functions, this quantity can be shown to be the sufficient statistics for decision-making.} This is illustrated in \cref{fig:two_stage}. The noiseless public feedback signal is 
\begin{align}
S \Equaldef \frac{1}{N} \sum_{i=1}^N A_i^{(1)}.
\end{align}

\subsection{Policies, utility, and Bayesian Nash Equilibrium}

The agents use policies on their observations to determine their actions. In the first stage, agent $i$'s policy relies exclusively on its private signal
\begin{align}
    A_i^{(1)} = \pi_i^{(1)}(Y_i).
\end{align}

In the second stage, the policy maps the initial private signal and the public feedback on first-stage actions to a final decision
\begin{align}
A_i^{(2)} = \pi_i^{(2)}(Y_i, S).
\end{align}

Denote the action of the $i$-th agent during the game by $A_i \Equaldef (A_i^{(1)}, A_i^{(2)})$, 
the action profile of all other agents $A_{-i} = (A_1, \ldots, A_{i-1}, A_{i+1}, \ldots A_N)$, and the fundamental $\Theta$, the utility function of agent $i\in \ags$ is defined as

\begin{multline}
u_i(A_i, A_{-i}, \Theta) \\ \Equaldef \big(A_i^{(1)} + \gamma A_i^{(2)}\big) \Bigg( \frac{1}{N}\sum_{j=1}^N \big(A_j^{(1)} + A_j^{(2)}\big) - \Theta \Bigg),
\end{multline}
where $\gamma\in(0,1)$ denotes the problem's discount factor.

\vspace{5pt}

\begin{remark}
The utility function is increasing in the aggregate of agents that take the risky action, and decreasing in the fundamental. Therefore, it has \textit{strategic complementarity} \cite{dahleh2016coordination}. Moreover, it is linear in the action profile.
\end{remark}

\vspace{5pt}

Denote the policy profile of the $i$-th agent during the game by $\pi_i = (\pi_i^{(1)}, \pi_i^{(2)})$ and the policy profile $\pi_{-i} = (\pi_1, \ldots, \pi_{i-1}, \pi_{i+1}, \ldots, \pi_N)$ adopted by all other agents. Let $\mathcal U_i$ denote the expected utility of agent~$i$, defined as
\begin{align}
    \mathcal U_i \big(\pi_i, \pi_{-i} \big) \Equaldef \mathbb E \big[ u_i (A_i, A_{-i}, \Theta) \big],
\end{align}
where
$A_i^{(1)} = \pi_i^{(1)}(Y_i)$ and $A_i^{(2)} = \pi_i^{(2)}(Y_i, S)$,
and $A_j^{(1)} = {\pi_j}^{(1)}(Y_j)$ and $A_j^{(2)} = {\pi_j}^{(2)}(Y_j, S)$ for ${j \neq i}$.

\vspace{5pt}

\begin{definition} \label{def:bayesian_nash_equilibrium}
A homogeneous policy profile $\pi^\star$ is a \emph{Bayesian Nash Equilibrium (BNE)} if
\begin{align}
\mathcal U_i \big(\pi^\star, \pi_{-i}^\star \big) 
\;\ge\; 
\mathcal U_i \big(\pi, \pi_{-i}^\star \big), 
\quad  i \in [N],
\end{align}
for any admissible policy $\pi$, where $\pi_{-i}^\star$ denotes the profile in which all agents other than $i$ adopt the policy $\pi^\star$.
\end{definition}

\section{Preliminaries on Finite-Population Model}
We focus on policies of the following form.

\vspace{5pt}

\begin{figure*}[t!]
\small
\begin{multline}\label{eq:first_stage_average_utility}
\mathbb E \!\left[
\frac{1}{N}
\Bigg(
\sum_{j \neq i}
\indic{Y_j \le \max\{\tau^\star, \lambda^\star(S)\}}
+ 1
\Bigg)
\!-\! \Theta
\, \Bigg|\, 
Y_i \!=\! y_i
\right]
\\ \ge 
\gamma
\mathbb E \Bigg[
\indic{Y_i \le \lambda^\star(S)}
\left(
\frac{1}{N}
\left(
\sum_{j \neq i}
\indic{Y_j \!\le\! \max\{\tau^\star, \lambda^\star(S)\}} \right)
\!+\! 1
\right)
\!-\! \Theta
\Bigg| \,
Y_i \!=\! y_i
\Bigg],
\end{multline}
\end{figure*}

\begin{definition}
The policies of \emph{homogeneous threshold type} are defined as
\begin{subequations}
\begin{align}
    \pi_i^{(1)} (Y_i) &\Equaldef  \indic{Y_i \leq \tau} \\
    \pi_i^{(2)} (Y_i, S) &\Equaldef  \indic{Y_i \leq \lambda(S)},
\end{align}
\end{subequations}
where $\mathbb{1}(\cdot)$ denotes the indicator function, $\tau$ is the first-stage threshold, and $\lambda(S)$ is the second-stage threshold, which depends on the first-stage participation $S$.
\end{definition}

\vspace{5pt}

\begin{lemma}\label{lemma:bne_policy}
The BNE  policy profile of homogeneous threshold type is characterized by two-stage thresholds $(\tau^\star, \lambda^\star)$, determined as follows\footnote{The existence of BNE policies of homogeneous threshold type is currently under investigation in a separate paper.}.

\paragraph{First stage} Agent $i$ takes the risky action in the first stage, i.e., $Y_i \le \tau^\star$, if and only if the conditional expected utility from acting in the first stage is no less than that from delaying to the second stage according to \eqref{eq:first_stage_average_utility}.

\paragraph{Second stage}
Agent $i$ who delayed the decision to the second stage chooses to take the risky action, i.e., $Y_i \le \lambda^\star(S)$, if and only if the conditional expected utility from acting is nonnegative:
\begin{multline}\label{eq:second_stage_average_utility}
\Bigg(1 - s - \frac{1}{N}\Bigg)
\mathbb{P} \Big(Y_j \le  \lambda^\star(s) \ \Big| \  Y_j > \tau^\star,\ Y_i = y_i,\ S = s\Big)  \\ +
\frac{1}{N} + s - \mathbb{E}\Big[\Theta \ \Big|\ Y_i = y_i,\ S = s\Big] \ge 0,
\end{multline}

\end{lemma}

\vspace{5pt}

\begin{proof}
Recall that, under the threshold-type policies, the total action taken by agent $i$ across the two stages satisfies
$A_i^{(1)} + A_i^{(2)}
=
\indic{Y_i \le \max\{\tau^\star, \lambda^\star(S)\}}$. 
We begin by analyzing the second stage.
Under a homogeneous policy profile characterized by $(\tau^\star,\lambda^\star)$, agent~$i$ takes the risky action if and only if the conditional expected utility of acting is greater than or equal to that of not acting, which yields zero utility, leading to \eqref{eq:second_stage}, where $\mathcal S$ denotes the set of agents who take action in the first stage, $\mathcal S = \{\, i \in [N] \mid A_i^{(1)} = 1 \,\}$.
This condition leads to \eqref{eq:second_stage_average_utility}.

\begin{figure*}
\small
\begin{multline}\label{eq:second_stage}
\frac{1}{N}
\left(
\sum_{j \notin \mathcal{S} \cup \{i\}}
\mathbb{E}\!\Big[
\indic{Y_j \le \max\{\tau^\star, \lambda^\star(S)\}}
\ \Big| \ Y_j > \tau^\star, Y_i = y_i, S = s
\Big] + 1
\right) 
+ s - \mathbb{E}\Big[\Theta \ \Big| \ Y_i = y_i, S = s\Big] \ge 0
\end{multline}
\end{figure*}

Similarly, for the first stage, the agent compares the conditional expected utility of acting immediately with that of delaying the decision to the second stage. Acting immediately yields a higher expected utility if and only if \eqref{eq:first_stage_average_utility} holds.
\end{proof}

\vspace{5pt}

\begin{theorem} \label{theorem:coordination_improvement}
    The option to delay facilitates participation
    when $\gamma \in (0, \bar{\gamma})$ for some $\bar{\gamma} > 0$. In particular,
\begin{equation}
\sum_{i=1}^N \indic{Y_i \le \tau_{\text{single}}^\star}
\;\le\;
\sum_{i=1}^N \indic{Y_i \le \max\{\tau^\star, \lambda^\star(s)\}}
\end{equation}
for all $s=k/N$, $k\in[N-1]$.
\end{theorem}

\vspace{5pt}

\begin{proof}
The proof is in Appendix~\ref{proof_thm1}.
\end{proof}

\vspace{5pt}

The above theorem shows that the option to delay can induce a larger number of agents to take the risky action. However, it is not immediate whether this increase in coordination translates into higher expected utility. To characterize when such improvements arise, we henceforth focus on the infinite-population two-stage global game, which admits a tractable characterization of equilibrium behavior. The following proposition establishes a link between the finite- and infinite-population settings.

\vspace{5pt}

\begin{proposition} \label{prop:infinite_population}
Fix a first-stage threshold $\tau$ and let ${\lambda_N^{\star}}$ denote the BNE second-stage policy
in a population of size $N$ (assuming it exists and it is unique). Suppose the limit $\lambda^\star (s) = \lim_{N \to \infty} \lambda_N^{\star} (s_N)$ is well-defined, where $s_N$ denotes the realization of the portion of agents acting in the first stage and satisfies $\lim_{N \to \infty} s_N = s$. In the population limit, the limiting policy profile $(\tau,\lambda^\star)$ is characterized as follows:

\begin{enumerate}[(a)]
\item The second-stage policy satisfies\footnote{Technically, $(\theta,\tau, \sigma,s)$ can be such that the inequality is exact. In that case, the lower limit and the upper limit to $s$ is not the same. However, this situation happens with probability zero.}:
\begin{align}
\textstyle\lambda^\star (s)
=
\begin{cases}
+\infty, & \text{if } \theta = \tau - \sigma \Phi^{-1}(s) \le 1 \\
-\infty, & \text{otherwise},
\end{cases}
\end{align}
where $\theta$ denotes the realization of the fundamental $\Theta$.

\item In equilibrium, agent~$i$ chooses to take the risky action in the first stage, i.e., $A_i^{(1)} = 1$, if and only if 
\begin{multline}
\mathbb{E}\Big[F(\Theta, S) - \Theta \mid Y_i = y_i\Big] \\
\shoveleft{\geq \gamma\, \mathbb{E}\Big[\indic{\Theta \le 1}\,\big(F(\Theta, S) - \Theta\big) \mid Y_i = y_i\Big],}
\end{multline}
where 
\begin{equation}
F(\Theta, S) \Equaldef \Phi\left(\frac{\max\{\tau, \lambda^\star(S)\} - \Theta}{\sigma}\right).
\end{equation}
\end{enumerate}
\end{proposition}

\vspace{5pt}

\begin{proof}
The proof is in Appendix~\ref{proof:prop:infinite_population}.
\end{proof}

\section{Infinite-Population Model}

The equilibrium analysis becomes significantly more tractable in the limit of a large number of agents.

\subsection{First-stage threshold policy}
We assume that the agents utilize a homogeneous threshold policy in the first stage. That is, agent $i$ takes the risky action early if and only if their private signal $Y_i$ is below a common critical threshold $\tau$, i.e., $A_i^{(1)} = \indic{Y_i \leq \tau}.$

Conditioned on the realization of the fundamental state $\Theta = \theta$, the fraction of agents who take the risky action in the first stage is
given by
\begin{equation}
s = \lim_{N\rightarrow \infty}\frac{1}{N} \sum_{i=1}^N \indic{Z_i \leq \tau - \theta} \overset{(a)}{=} \mathbb{P} \big(Z_i \leq \tau -\theta \big), 
\end{equation}
where $Z_i \sim \mathcal{N}(0,\sigma^2)$ and $(a)$ follows from the strong law of large numbers for exchangeable processes (see e.g.\ Example 1, Chapter 7.3., \cite{chow2003probability}). Therefore, under the first-stage threshold policy, the realized variable $S=s$ corresponds  to  
$s = \Phi\left(\frac{\tau - \theta}{\sigma}\right)$,
where $\Phi(\cdot)$ denotes the CDF of the standard Gaussian distribution.

Since at the end of the first stage, each agent who did not take the risky action observes the fraction of agents who acted in the first stage, and since, the standard normal CDF $\Phi(\cdot)$ is strictly increasing and continuous, it is invertible, allowing to recover the fundamental $\theta$ perfectly. Consequently, using this information, each such agent can compute 
\begin{equation}
\tau - \sigma \Phi^{-1}(s)=\theta.
\end{equation}
In the infinite-population regime, this implies that all agents who delay taking the risky action observe the true realization $\Theta = \theta$ prior to making their second-stage decision.

\vspace{5pt}

\begin{remark}
The exact recovery of the fundamental state by the remaining agents arises from signaling through infinitely many independent binary channels. While each individual agent's action is based on a noisy private signal, aggregating these actions over an infinite population conveys an unbounded amount of information\footnote{In the sense of Shannon's Information Theory \cite{cover2006elements}.} from the first stage to the second. As a result, the continuous parameter $\theta$ can be recovered perfectly. This mechanism does not arise in the finite-population setting.
\end{remark}

\vspace{5pt}

Because the fundamental is perfectly revealed, the second-stage policy for any agent who has not yet taken the risky action reduces to a deterministic mapping from $\theta$ to an action. Consequently, the initial noisy private signal $Y_i$ becomes irrelevant for the remaining decision. This is analogous to sequential social learning, where private signals can be disregarded without loss of optimality once a sufficient number of prior actions has been observed \cite{Chamley2004RationalHerdsDOI,acemoglu2011bayesian}. Henceforth, we denote the total fraction of the population taking the risky action across both stages by {$F(\theta)$}.

\subsection{Optimal second-stage policy}
In the second stage, the fundamental state $\Theta = \theta$ is perfectly revealed to all agents. The mass of agents who already committed to the risky action in the first stage is given by 
\begin{equation}
F^{(1)}(\theta) = \lim_{N \to \infty} \frac{1}{N} \sum_{i=1}^N A_i^{(1)} = \Phi \left( \frac{\tau - \theta}{\sigma} \right).
\end{equation}
The remaining agents, which have a total mass of $1 - F^{(1)} (\theta)$, now play a perfect-information coordination game.

If an agent takes the risky action in the second stage, their payoff is $\gamma\big(F(\theta) - \theta\big)$. Because the state is known, we can identify three distinct strategic regions based on $\theta$:

 \subsubsection{Upper Dominance Region ($\theta \geq 1$)} Even if all remaining agents take the risky action, the maximum possible aggregate action is $F(\theta) = 1$. The payoff would be $\gamma(1 - \theta) \leq 0$. The risky action is strictly dominated, so the remaining agents take the safe action: $A_i^{(2)} = 0$.
 
     \subsubsection{Lower Dominance Region \big($\theta < F^{(1)}(\theta)=s$\big)} The payoff is strictly positive even if no additional agents take the risky action, because the first-stage alone guarantees $F^{(1)}(\theta) - \theta > 0$. The risky action is strictly dominant, and all remaining agents take the risky action: $A_i^{(2)} = 1$.
     \subsubsection{Multiplicity Region \big($s=F^{(1)}(\theta) \leq \theta < 1$\big)} The remaining agents play a pure coordination game. If they all take the risky action, the total active population becomes $F(\theta) = F^{(1)}(\theta) + F^{(2)}(\theta) = 1$, yielding a positive payoff $\gamma\big(1 - \theta\big) > 0$. If all other agents take the safe action, the total active population remains $F^{(1)}(\theta)$, yielding a negative payoff $\gamma(F^{(1)}\big(\theta) - \theta\big) \leq 0$. We assume agents coordinate on the Pareto-dominant equilibrium\footnote{An equilibrium is considered Pareto dominant if there exists no other equilibrium profile that yields a strictly higher expected payoff for at least one agent without decreasing the payoff of any other agent.} where the coordination succeeds. Thus, all remaining agents take the risky action: $A_i^{(2)} = 1$.

Combining these regions, the optimal second-stage policy for the agents who chose delay is an indicator function of $s$
\begin{align}
\textstyle {\pi_i^\star}^{(2)}(s)
=\indic{s\geq\Phi\!\left(\frac{\tau-1}{\sigma}\right)}, \ \ i\in\ags.
\end{align}
As a result of using ${\pi_i^\star}^{(2)}(s)$, the total \textit{ex-post} aggregate action becomes
\begin{align}
\textstyle F(\theta) = 
\begin{cases} 
1, & \text{if } \theta \leq 1 \\ 
\Phi\left(\frac{\tau - \theta}{\sigma}\right), & \text{if } \theta > 1.
\end{cases}
\end{align}
Note that this is consistent with Proposition~\ref{prop:infinite_population}.

\subsection{Equilibrium condition}
We now use backward induction to evaluate the optimal first-stage decision. An agent~$i$ observing a private signal $Y_i = y$ must weigh the expected utility of taking the risky action in the first stage against the expected utility of delaying the decision. Define the function
\begin{align}\label{eq:indif}
\Delta(y) \Equaldef \mathbb{E}\Big[(1 - \gamma \indic{\Theta < 1}\big) \big(F(\Theta) - \Theta\big) \mid Y_i = y\Big].
\end{align} 
Note that this expression is the utility difference between taking the risky action in the first stage and the second stage, given the information available at the first stage. As discussed in Proposition~\ref{prop:infinite_population}, for the proposed threshold $\tau^\star$ to constitute a BNE policy profile, an agent must be \textit{indifferent} between taking the risky action in the first stage and delaying the decision to the second stage. This yields the \textit{indifference condition} $\Delta(\tau^\star)=0$.

To establish the existence of an equilibrium threshold policy, we proceed as 
follows. First, we notice that $\Delta(\tau) > 0$ as $\tau \to -\infty$ and 
$\Delta(\tau) < 0$ as $\tau \to +\infty$. By continuity of $\Delta(\tau)$, there 
must exist at least one value $\tau^\star$ such that $\Delta(\tau^\star) = 0$. 
However, uniqueness cannot be guaranteed unless we impose a condition on the noise 
variance.

\vspace{5pt}

\begin{theorem} \label{theorem:unique_BNE}
In the infinite population regime, suppose that noise variance satisfies ${\sigma^2 < 2\pi}$, then $\Delta(\tau)$ in~\eqref{eq:indif} is strictly decreasing function of $\tau$. As a result, there exists a unique threshold $\tau^\star$ 
that characterizes the BNE, i.e., $\Delta(\tau^\star) = 0$.
\end{theorem}
\vspace{5pt}

\begin{proof}
Per the discussion above, a (homogeneous) BNE for the infinite population game would be the pair ${\big(\indic{Y_i\leq \tau^\star},\indic{\Theta\leq 1}\big)}$ for the two stage actions, where the threshold $\tau^\star$ satisfies $\Delta(\tau^\star)=0$. Note that 
\begin{equation}
\alpha=\frac{1}{1+\sigma^2},
\qquad
\Theta\mid Y=\tau \sim \mathcal{N}(\alpha\tau,\alpha\sigma^2).
\end{equation}
In other words, $\Theta_\tau \Equaldef (\Theta\mid Y=\tau)$ can be written as
\begin{equation}
\Theta_\tau=\alpha\tau+\sigma\sqrt{\alpha}\,Z,
\qquad Z\sim\mathcal{N}(0,1).
\end{equation}
Note that $Z$'s distribution is independent of $\tau$. 

Define
\begin{equation}
g(\tau,\theta) \Equaldef
\begin{cases}
(1-\gamma)(1-\theta), & \theta\le 1,\\[4pt]
\Phi\!\left(\dfrac{\tau-\theta}{\sigma}\right)-\theta, & \theta>1.
\end{cases}
\end{equation}
Then, $\Delta(\tau)=\mathbb{E}\bigl[g(\tau,\Theta_\tau)\bigr].$ For almost every $Z$, differentiating with respect to $\tau$ gives
\begin{equation}
\frac{d}{d\tau}g(\tau,\Theta_\tau)= -\alpha(1-\gamma)
\end{equation}
for $\Theta_\tau<1$, and for $\Theta_\tau>1$,
\begin{equation}
\frac{d}{d\tau}g(\tau,\Theta_\tau)=
\frac{1}{\sigma}\phi\!\left(\frac{\tau-\Theta_\tau}{\sigma}\right)
+\alpha\left(-1-\frac{1}{\sigma}\phi\bigg(\frac{\tau-\Theta_\tau}{\sigma}\right)\bigg).
\end{equation}
Since $1-\alpha=\alpha\sigma^2$, this simplifies to
\begin{equation}
\frac{d}{d\tau}g(\tau,\Theta_\tau)=
\begin{cases}
-\alpha(1-\gamma), & \Theta_\tau<1,\\[6pt]
\alpha\!\left(\sigma\phi\!\left(\dfrac{\tau-\Theta_\tau}{\sigma}\right)-1\right), & \Theta_\tau>1.
\end{cases}
\end{equation}
Note that $-\alpha(1-\gamma)<0$, and using $\phi(x)\le\phi(0)=1/\sqrt{2\pi}$, we obtain
\begin{equation}
\alpha\left(\frac{\sigma}{\sqrt{2\pi}}-1\right) < 0
\ \ \text{whenever } \sigma<\sqrt{2\pi}.
\end{equation}
Hence, letting $\gamma=\max\!\left(-(1-\gamma),\,\frac{\sigma}{\sqrt{2\pi}}-1\right)<0$, we have
\begin{equation}
\Delta'(\tau)=\mathbb{E}\!\left[\frac{d}{d\tau}g(\tau,\Theta_\tau)\right]\le\gamma<0.
\end{equation}
Since $\Delta'(\tau)\le\gamma$, it follows that $\Delta(\tau)\to-\infty$ as $\tau\to+\infty$ and $\Delta(\tau)\to+\infty$ as $\tau\to-\infty$. Hence, by continuity of $\Delta(\tau)$, it has a root $\tau=\tau^\star$, which is unique by strict monotonicity.
\end{proof}

\vspace{5pt}

One of the implications of the above result is that the optimal threshold $\tau^\star$ is a decreasing function of the discount factor $\gamma$, which is a natural condition to appear: when $\gamma$ is small, the agents have more incentive to act on the first stage. 

\vspace{5pt}

\begin{proposition}
If $\sigma^2<2\pi$, the equilibrium threshold $\tau^\star$ for taking a risky action on the 
first stage, is strictly decreasing in the discount factor 
$\gamma$. Therefore, an increase in the discount factor strictly reduces the 
fraction of agents taking the risky action on the first-stage for any given realization of the 
fundamental $\Theta=\theta$.
\end{proposition}

\vspace{5pt}

\begin{proof}
The equilibrium threshold $\tau^\star$ is uniquely defined by the indifference 
condition $\Delta(\tau^\star; \gamma, \sigma) = 0$. By the Implicit Function 
Theorem, the derivative of the threshold with respect to $\gamma$ is given by
\begin{equation} \label{eq:IFT}
\frac{d\tau^\star}{d\gamma} = - \frac{\frac{\partial \Delta}{\partial \gamma}}
{\frac{\partial \Delta}{\partial \tau}}.
\end{equation}
From the monotonicity established in the proof of Theorem~\ref{theorem:unique_BNE}, we know that if 
$\sigma^2<2\pi$, then $\frac{\partial \Delta}{\partial \tau} (\tau)<0$. It remains to analyze the partial 
derivative of $\Delta$ with respect to $\gamma$. Notice that $\gamma$ only 
affects the agent's payoff when $\theta < 1$, where the net payoff of taking a 
risky action in the first stage versus delaying is $(1 - \gamma)(1 - \theta)$. 
Differentiating $\Delta$ with respect to $\gamma$ gives
\begin{equation} 
\frac{\partial \Delta}{\partial \gamma} = -\int_{-\infty}^{1} (1 - \theta) 
f(\theta \mid Y=\tau^\star) \, d\theta < 0,
\end{equation}
where the inequality holds since $\theta < 1$. From 
\eqref{eq:IFT}, we obtain $\frac{d\tau^\star}{d\gamma} < 0$. Since the fraction of agents taking the risky action on the first-stage  given $\Theta = \theta$ is $\mathbb{P}(Y \leq \tau^\star 
\mid \Theta = \theta) = \Phi\left(\frac{\tau^\star - \theta}{\sigma}\right)$, 
which is strictly increasing in $\tau^\star$, an increase in $\gamma$ strictly 
reduces this probability for any given $\theta$.
\end{proof}

\section{Coordination Efficiency}

To compare the efficiency of the system under two stages vs. single stage homogeneous threshold policies, we consider a metric based on the aggregate expected payoff (i.e., welfare) of all agents. The larger this collective payoff, the more efficient the agents are in terms of coordinating their actions towards performing a task with difficulty $\theta$.

Recall that under a first-stage threshold policy with parameter $\tau$, the first-stage adoption rate as a function of the fundamental $\theta$ is given by
$F^{(1)}(\theta) = \Phi\!\left(\frac{\tau - \theta}{\sigma}\right)$.
Under Pareto-dominant equilibrium selection in the second stage, where agents observe $\theta$ perfectly through the public signal $s$, all remaining agents take the risky action if and only if $\theta \leq 1$. This yields a second-stage adoption rate of
$\bigl(1 - F^{(1)}(\theta)\bigr)\,\mathbb{1}(\theta \leq 1)$. 

To quantify the efficiency of coordination under $\tau$, we define the welfare as
\begin{multline}
W_{\text{two-stage}}(\tau, \sigma, \gamma)  
\Equaldef \int_{1}^{\infty} \underbrace{F^{(1)}(\theta)\bigl(F^{(1)}(\theta) - \theta\bigr)}_{\text{aggregate payoff when $\theta>1$}}\,\phi(\theta)\,d\theta\\ +\int_{-\infty}^{1} \underbrace{\Big[F^{(1)}(\theta) + \gamma\big(1 - F^{(1)}(\theta)\big)\Big](1 - \theta)}_{\text{aggregate payoff when $\theta\leq1$ }}\,\phi(\theta)\,d\theta,
\end{multline}
where $\phi$ denotes the standard Gaussian density. Figure~\ref{fig:welfare} illustrates the welfare of the two-stage global game as a function of the threshold $\tau$.

\begin{figure}[t!]
    \centering
\includegraphics[width=0.8\linewidth]{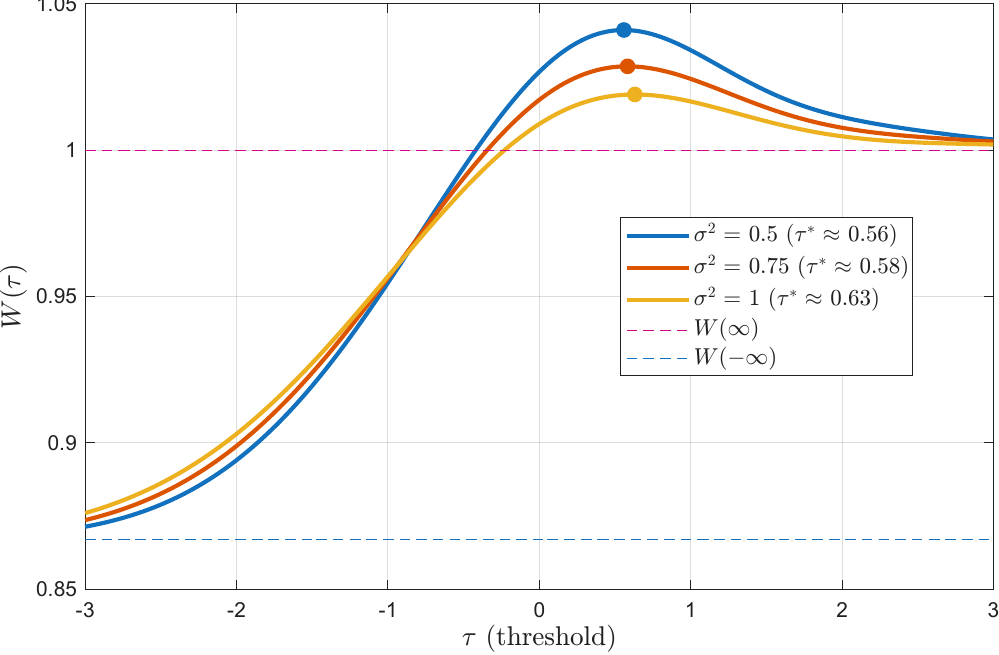}
    \caption{Per agent welfare for a two-stage global game with $\gamma=0.8$.}
    \label{fig:welfare}
\end{figure}

To assess whether the option to delay is beneficial or detrimental in equilibrium, we compare the welfare of the two-stage global game against that of the single-stage global game. The welfare for the single-stage game is given by
\begin{align}
    W_{\text{single-stage}}(\tau, \sigma) \Equaldef \int_{-\infty}^{\infty} F^{(1)} (\theta)
    \big(F^{(1)} \big(\theta) - \theta\big) \phi(\theta)\,d\theta.
\end{align}

For the single-stage global game, the indifference condition that determines the BNE threshold is given by \cite{Vasconcelos:2023}:
\begin{align}\label{eq:indif_single}
    \mathbb{E}\Big[ F^{(1)}(\Theta) - \Theta \mid Y_i = \tau^\star \Big] = 0.
\end{align}

Let $\tau_{\mathrm{single}}^\star(\sigma)$ denote the BNE threshold of the single-stage global game and $\tau^\star(\sigma,\gamma)$ the BNE threshold of the two-stage global game. The \emph{value of the second stage} is defined as
\begin{multline} \label{eq:second_stage_value}
    V(\sigma,\gamma) \Equaldef W_{\text{two-stage}}\bigl(\tau^\star(\sigma,\gamma),\, \sigma,\, \gamma\bigr) \\ - W_{\text{single-stage}}\bigl(\tau_{\text{single}}^\star(\sigma),\, \sigma\bigr),
\end{multline}
where each welfare is evaluated at the respective equilibrium threshold of its own game. When $V(\sigma,\gamma) > 0$, the two-stage game achieves more welfare than the single-stage game: the safety net provided by the noiseless second stage more than compensates for the strategic delay it induces. When $V(\sigma,\gamma) < 0$, the opposite holds: more agents tend to postpone their decisions to a costly second stage, and the resulting expected coordination loss exceeds the benefit of taking the risky action.

\begin{figure}[t!]
    \centering
    \includegraphics[width=0.85\linewidth]{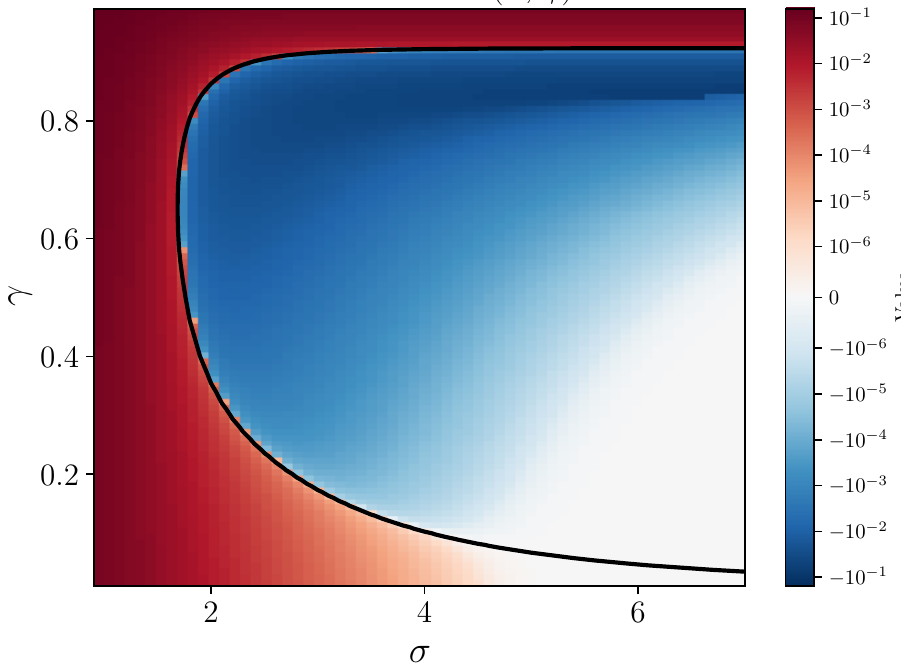}
    \caption{Illustration of the function $V(\sigma,\gamma)$ defined in \eqref{eq:second_stage_value}. The red region indicates the parameter regime $(\sigma,\gamma)$ in which the option to delay is beneficial, whereas the blue region corresponds to regimes in which it is detrimental.}
    \label{fig:second_stage_value}
\end{figure}

\subsection{When does the option to delay improve efficiency?} 
Figure~\ref{fig:second_stage_value} illustrates the regions in which the option to delay is beneficial. Interestingly, the option to delay can also be detrimental. The following theorem formalizes this observation.

\vspace{5pt}

\begin{theorem} \label{theorem:coordination_efficiency}
Consider a two-stage global game with infinitely many agents, noisy signals with variance $\sigma^2$ and discount factor $\gamma\in(0,1)$. The following hold \begin{enumerate}[(1)]    
    \item For any given noise level $\sigma > 0$, there exists $\bar{\gamma}_\sigma > 0$ such that welfare improves for $\gamma \in (0, \bar{\gamma}_\sigma)$.
    \item For any fixed discount factor $\gamma \in (0,1)$, there exists $\bar{\sigma}_\gamma > 0$ such that welfare improves for ${\sigma \in (0, \bar{\sigma}_\gamma)}$.
\end{enumerate}

\end{theorem}

\vspace{5pt}

\begin{proof}
The proof is in Appendix~\ref{proof:theorem:coordination_efficiency}.
\end{proof}

\vspace{5pt}

\begin{remark}
An important observation is that the equilibrium threshold $\tau^\star(\sigma, \gamma)$ does not, in general, maximize welfare in the two-stage game, i.e.,
\begin{equation}
W_{\text{two-stage}}\big(\tau^\star(\sigma, \gamma), \sigma, \gamma\big) 
< 
\max_{\tau \in \mathbb R} W_{\text{two-stage}}(\tau, \sigma, \gamma).
\end{equation}
As a result, it may occur that
\begin{multline}
\textstyle W_{\text{two-stage}}\big(\tau^\star(\sigma, \gamma), \sigma, \gamma\big) 
 < 
W_{\text{single-stage}}(\tau_{\text{single}}^\star (\sigma), \sigma) 
\\< 
\max_{\tau \in \mathbb R} W_{\text{two-stage}}(\tau, \sigma, \gamma),    
\end{multline}
which shows there is a gap between equilibrium and welfare-optimal coordination.
\end{remark}

\section{Conclusions and future work}

This work considers a sequential stochastic coordination game with imperfect observations among $N$ agents. Within the class of homogeneous threshold policies, we provided several preliminary results that showed that adding the option to delay taking a risky action and reducing an agent's payoff may be beneficial. When the system has a finite number of agents, we have shown that adding a second stage increases the probability an agent will take the risky action in the first stage in equilibrium. When the number of agents is asymptotically large, we show there exist regions of the $(\gamma,\sigma)$ parameter space where having a second stage is strictly better than not having one. Therefore, the feedback mechanism constitutes an important tool for improving coordination. Future work on this topic includes the characterization of information theoretic limits, learning algorithms, and the impact of social network structure on coordination. 

\bibliography{ref}
\bibliographystyle{ieeetr}

\appendix

\subsection{Proof of Theorem~\ref{theorem:coordination_improvement}}\label{proof_thm1}

Let $(\tau^\star, \lambda^\star)$ denote the BNE threshold policy for the two-stage global game, and let $\tau_{\text{single}}^\star$ denote the BNE threshold in the corresponding single-stage global game. Since $\tau^\star \le \max\{\tau^\star, \lambda^\star(s)\}$, it suffices to show that $\tau_{\text{single}}^\star \le \tau^\star$. In the single-stage global game, under the BNE policy $A_i^{(1)} = \indic{Y_i \le \tau_{\text{single}}^\star}$, agent $i$ chooses to take the risky action if and only if
\begin{equation}
\mathbb{E}\!\left[
\frac{1}{N}\left(\sum_{j \neq i} \indic{Y_j \le \tau_{\text{single}}^\star} + 1\right) - \Theta
\ \Bigg|\ Y_i = y_i
\right] \ge 0.
\end{equation}

Based on \eqref{eq:first_stage_average_utility}, we define the net gain from acting in the first stage of the two-stage game evaluated at the policy $(\tau_{\text{single}}^\star, \lambda^\star)$ as $\Delta^N(y_i)$ defined in \eqref{eq:net_gain}.
\begin{figure*}
\footnotesize
\begin{multline}\label{eq:net_gain}
\Delta^N(y_i) \Equaldef
\mathbb{E}\!\left[
\frac{1}{N}\left(\sum_{j \neq i} \indic{Y_j \le \max\{\tau_{\text{single}}^\star, \lambda^\star(S)\}} + 1\right)  - \Theta
\ \Bigg|\ Y_i = y_i
\right] 
\\- \gamma\,
\mathbb{E}\!\left[
\indic{Y_i \le \lambda^\star(S)}
\left(\frac{1}{N}\left(\sum_{j \neq i} \indic{Y_j \le \max\{\tau_{\text{single}}^\star, \lambda^\star(S)\}} + 1\right) - \Theta\right)
\ \Bigg|\ Y_i = y_i
\right].
\end{multline}
\end{figure*}
We now show that $\Delta^N(y_i) > 0$ for all $y_i \le \tau_{\text{single}}^\star$, by considering three cases.

\smallskip\noindent
\emph{(i) Interior region:} Suppose $y_i \in [\min_s \lambda^\star(s),\, \max_s \lambda^\star(s)]$. For sufficiently small $\gamma$, the inequality in \eqref{eq:small_gamma} holds.
\begin{figure*}
\footnotesize
\begin{multline}\label{eq:small_gamma}
\Delta^N(y_i)
\ge
\mathbb{E}\!\left[
\frac{1}{N}\left(\sum_{j \neq i} \indic{Y_j \le \max\{\tau_{\text{single}}^\star, \lambda^\star(S)\}} + 1\right) - \Theta
\ \Bigg|\ Y_i = y_i
\right]  \\
\quad - \gamma \sup_{y_i}\,
\mathbb{E}\!\left[
\frac{1}{N}\left(\sum_{j \neq i} \indic{Y_j \le \max\{\tau_{\text{single}}^\star, \lambda^\star(S)\}} + 1\right) - \Theta
\ \Bigg|\ Y_i = y_i
\right]  \\
>
\mathbb{E}\!\left[
\frac{1}{N}\left(\sum_{j \neq i} \indic{Y_j \le \tau_{\text{single}}^\star} + 1\right) - \Theta
\ \Bigg|\ Y_i = y_i
\right]
\ge 0.
\end{multline}
\end{figure*}

\smallskip\noindent
\emph{(ii) Lower region:} Suppose $y_i < \min_s \lambda^\star(s)$. Then $\indic{Y_i \le \lambda^\star(S)} = 1$ almost surely, so the $\gamma$ term vanishes from $\Delta^N$ and we obtain \eqref{eq:lower_region}.
\begin{figure*}
\footnotesize

\begin{multline}\label{eq:lower_region}
\Delta^N(y_i)
=
\mathbb{E}\!\left[
\frac{1}{N}\left(\sum_{j \neq i} \indic{Y_j \le \max\{\tau_{\text{single}}^\star, \lambda^\star(S)\}} + 1\right) - \Theta
\ \Bigg|\ Y_i = y_i
\right] 
>
\mathbb{E}\!\left[
\frac{1}{N}\left(\sum_{j \neq i} \indic{Y_j \le \tau_{\text{single}}^\star} + 1\right) - \Theta
\ \Bigg|\ Y_i = y_i
\right]
\ge 0.
\end{multline}
\end{figure*}

\smallskip\noindent
\emph{(iii) Upper region:} Suppose $y_i > \max_s \lambda^\star(s)$. Then $\indic{Y_i \le \lambda^\star(S)} = 0$ almost surely, so \eqref{eq:upper_region} holds.
\begin{figure*}
\footnotesize
\begin{multline} \label{eq:upper_region}
\Delta^N(y_i)
= 
(1 - \gamma)\,
\mathbb{E}\!\left[
\frac{1}{N}\left(\sum_{j \neq i} \indic{Y_j \le \max\{\tau_{\text{single}}^\star, \lambda^\star(S)\}} + 1\right) - \Theta
\ \Bigg|\ Y_i = y_i
\right]  
\\ >
(1 - \gamma)\,
\mathbb{E}\!\left[
\frac{1}{N}\left(\sum_{j \neq i} \indic{Y_j \le \tau_{\text{single}}^\star} + 1\right) - \Theta
\ \Bigg|\ Y_i = y_i
\right]
\ge 0.
\end{multline}
\end{figure*}

\smallskip
Therefore, $\Delta^N(y_i) > 0$ for all $y_i \le \tau_{\text{single}}^\star$. By Lemma~\ref{lemma:bne_policy}, the policy $(\tau_{\text{single}}^\star, \lambda^\star)$ cannot be a BNE, which implies that  $\tau_{\text{single}}^\star < \tau^\star$. Consequently, the option to delay increases the  first-stage equilibrium threshold, which implies  that more agents take the risky action in the first stage.\hfill\QED

\subsection{Proof of Proposition~\ref{prop:infinite_population}} \label{proof:prop:infinite_population}
To prove the proposition, we need the following lemma.
\begin{lemma} \label{lemma:first_stage_policy_limit}
Let $s^N=\frac{1}{N}\sum_{i=1}^N A_i^{(1)}$ and suppose that the first-stage policy satisfies
$A_i^{(1)}=\indic{y_i\le \tau}$.
Then, for any bounded continuous function $g$, it holds that
\begin{align}
    \textstyle \lim_{N\to\infty} \mathbb E[g(\Theta)\mid Y_i=y_i, S^N=s^N] = g(\theta),
\end{align}
where $\theta=\tau-\sigma\Phi^{-1}(s)$ and $s=\lim_{N\to\infty} s^N$.
\end{lemma}
\begin{proof}
Conditioned on $\Theta=\theta$, the first-stage action satisfies
$\textstyle \mathbb P(A_i^{(1)}=1\mid\Theta=\theta)
=
\Phi\!\left(\frac{\tau-\theta}{\sigma}\right)$.
Hence $S^N$ is the sample mean of i.i.d. Bernoulli random variables with parameter
$\textstyle p(\theta)=\Phi\!\left(\frac{\tau-\theta}{\sigma}\right)$.

By the law of large numbers, we have
$ S^N
\;\xrightarrow[]{a.s.}\;
p(\theta)
=
\Phi\!\left(\frac{\tau-\theta}{\sigma}\right)$.
Therefore, the observation of $S^N$ asymptotically reveals $p(\theta)$, and since $\Phi$ is strictly increasing,
it holds that $\theta = \tau-\sigma\Phi^{-1}(s)$.

Using Bayes' rule,
\begin{align}
    &\mathbb E [ g(\Theta) \,|\, Y_i = y_i, S^N = s^N ] \nonumber \\
    &= \int g(\theta') \, f (\theta' \,|\, Y_i = y_i, S^N \!=\! s^N) \, \mathrm d \theta' \nonumber \\
    &= \int g(\theta') \, \frac{f (S^N \!=\! s^N \,|\, \Theta = \theta') f (\theta' \,|\, Y_i = y_i)}{f (S^N = s^N \,|\, Y_i = y_i )} \, \mathrm d \theta'.
\end{align}

The likelihood of $S^N$ satisfies
\begin{equation}
f(S^N=s^N\mid\Theta=\theta')
\propto
\exp\!\left(-N D(s^N\|p(\theta'))\right),
\end{equation}
where $D(\cdot\|\cdot)$ denotes the Kullback--Leibler divergence.  
As $N\to\infty$, this likelihood concentrates at the unique point satisfying
$s=p(\theta')$.

Hence the posterior distribution of $\Theta$ converges weakly to the Dirac measure $\delta (\theta' - \theta)$,
where $\theta=\tau+\sigma\Phi^{-1}(s)$. Therefore,
\begin{multline}
\textstyle \lim_{N\to\infty}
\mathbb E[g(\Theta)\mid Y_i=y_i, S^N=s^N]
\\ \textstyle =
\int g(\theta')\,\delta (\theta'-\theta) \, \mathrm d\theta'
=
g(\theta).    
\end{multline}
\end{proof}

\textit{Proof of the proposition:}
We express
$s^N = \Phi\!\left(\frac{\tau-\theta}{\sigma}\right) + \delta^N$,
where $\delta^N = \frac{1}{N}\sum_{i=1}^{N} A_i^{(1)} - \Phi\!\left(\frac{\tau-\theta}{\sigma}\right)$.
Hence, we obtain 
$ \theta
=
\tau
-
\sigma
\Phi^{-1}(s^N-\delta^N)$.

\paragraph*{Second stage}

Consider the conditional expectation of the second-stage utility in a population of size $N$: Using conditional independence of $Y_j$ given $\Theta$, we obtain
\begin{align}
&\textstyle 
\mathbb E\!\Big[
\frac{1}{N}
\left(
\sum_{j\notin\mathcal S\cup\{i\}}
\indic{Y_j\le \max\{\tau, \lambda^N(S^N)\}}
+1
\right) \nonumber \\
&\qquad +S^N-\Theta
\ \Big|\ 
Y_i=y_i,\ S^N = s^N
\Big] \nonumber\\
&=
\textstyle (1 \!-\! s^N \!-\! \frac{1}{N}) 
\mathbb P\!\Big[ Y_j\le\max \{\tau, \lambda^N(S^N)\} \nonumber \\
&\textstyle \Big| Y_j \!>\! \tau, Y_i\!=\! y_i, S^N \!=\! s^N
\Big] \!+\! \frac{1}{N} \!+\! s^N
\!-\!
\mathbb E[\Theta | Y_i\!=\!y_i,S^N\!=\!s^N]\nonumber \\
&=
\textstyle (1 \!-\! s^N \!-\! \frac{1}{N}) \frac{\mathbb E\!\left[ \Phi\!\left(\frac{\max \{\tau, \lambda^N(s^N)\} - \Theta}{\sigma}\right) - \Phi\!\left(\frac{\tau - \Theta}{\sigma}\right) \Big| Y_i\!=\!y_i,S^N \!=\! s^N \right]}{\mathbb E\!\left[ 1 - \Phi\!\left(\frac{\tau - \Theta}{\sigma}\right) \Big| Y_i\!=\!y_i,S^N \!=\! s^N \right]}\nonumber \\
&\textstyle \qquad
+ \frac{1}{N} + s^N
-
\mathbb E[\Theta \,|\, Y_i=y_i, S^N=s^N],
\end{align}
where $\mathcal S^N$ denotes the set of agents who take action in the first stage,
$\mathcal S^N = \{\, i \in [N] \mid A_i^{(1)} = 1 \,\}$.

Let
$s=\lim_{N\to\infty}s^N$ and $\lambda(s)=\lim_{N\to\infty}\lambda^N(s^N)$. As $N$ tends to infinity, the law of large numbers implies $\delta^N\to0$. Consequently,
\begin{equation}
\theta
=
\tau-\sigma\Phi^{-1}(s).
\end{equation}

Furthermore, Lemma~\ref{lemma:first_stage_policy_limit} implies that
\begin{align}
&\textstyle (1 \!-\! s^N \!+\! \frac{1}{N}) 
\frac{\mathbb E\!\left[ \Phi\!\left(\frac{\max \{\tau, \lambda^N(S^N)\} - \Theta}{\sigma}\right) - \Phi\!\left(\frac{\tau - \Theta}{\sigma}\right) \Big| Y_i\!=\!y_i, S^N \!=\! s^N \right]}{\mathbb E\!\left[ 1 - \Phi\!\left(\frac{\tau - \Theta}{\sigma}\right) \Big| Y_i\!=\!y_i, S^N \!=\! s^N \right]} \nonumber \\
&\textstyle \qquad
+ \frac{1}{N} + s^N
-
\mathbb E[\Theta \,|\, Y_i=y_i, S^N=s^N] \nonumber
\\
&
\textstyle \xrightarrow[]{N \to \infty}
(1-s)
\frac{\Phi\!\left(\frac{\max \{\tau, \lambda(s)\} - \theta}{\sigma}\right) - \Phi\!\left(\frac{\tau - \theta}{\sigma}\right)}{1 - \Phi\!\left(\frac{\tau - \theta}{\sigma}\right)}
+
s
-
\theta .
\end{align}

The best response in the second stage therefore requires
\begin{equation}
\textstyle (1-s)\frac{\Phi\!\left(\frac{\max \{\tau, \lambda(s)\} - \theta}{\sigma}\right) - \Phi\!\left(\frac{\tau - \theta}{\sigma}\right)}{1 - \Phi\!\left(\frac{\tau - \theta}{\sigma}\right)}
+
s-\theta
\ge0 .
\end{equation}

Using $\theta=\tau-\sigma\Phi^{-1}(s)$, the above condition yields
\begin{equation}
\lambda(s)
=
\begin{cases}
\infty,
& \text{if }\theta\le1,\\[4pt]
-\infty,
& \text{otherwise}.
\end{cases}
\end{equation}

\paragraph*{First stage}

Consider the conditional expectation of the first-stage utility. Taking $N\to\infty$ and using the same argument as above, if $A_i^{(1)} = 1$, then
\begin{align*}
&\textstyle 
\mathbb E\!\Big[
\frac{1}{N}
\Big(
\sum_{j\neq i}
\indic{Y_j\le\max\{\tau, \lambda^N(S^N)\}}
+1
\Big)
-
\Theta
\ \Big|\ 
Y_i=y_i
\Big]
\\
&\textstyle = \mathbb E \left[ \frac{1}{N} \left( \sum_{j \neq i} \Phi \left( \frac{\max \{\tau, \lambda^N(S^N)\} - \Theta}{\sigma} \right) + 1 \right) - \Theta \,\Big|\, Y_i = y_i \right] \\
&\textstyle
\xrightarrow[]{N \to \infty}
\mathbb E\!\left[
\Phi\!\left(\frac{\max \{\tau, \lambda(S)\}-\Theta}{\sigma}\right)
-
\Theta
\ \Big|\ 
Y_i=y_i
\right].
\end{align*}

Otherwise, if $A_i^{(1)} = 0$, we obtain
\begin{align*}
&\textstyle \gamma \mathbb E \Big[ \indic{Y_i \!\leq\! \lambda^N(S^N)} \\
    &\textstyle \quad \times \left( \frac{1}{N} \left( \sum_{j \neq i} \Phi \left( \frac{\max \{\tau, \lambda^N(S^N)\} - \Theta}{\sigma} \right) \!+\! 1 \right) \!-\! \Theta \right) \,\Big|\, Y_i =\! y_i \Big] \\
&\textstyle
\xrightarrow[]{N \to \infty}
\gamma
\mathbb E\!\Big[
\underbrace{\indic{Y_i \!\leq\! \lambda(S)}}_{=\indic{\Theta \leq 1}}
\left(
\Phi\!\left(\frac{\max \{\tau, \lambda(S)\}-\Theta}{\sigma}\right)
\!-\!
\Theta
\right)
\!\Big| Y_i \!=\! y_i
\Big].
\end{align*}

This expression characterizes the limiting expected utility difference between acting in the first stage and delaying to the second stage. Therefore, the equilibrium first-stage threshold $\tau$ is determined by the indifference condition that the above expression equals zero at $Y_i=\tau$.
\hfill\QED

\subsection{Proof of Theorem~\ref{theorem:coordination_efficiency}} \label{proof:theorem:coordination_efficiency}

\textbf{Part~I.}
Recall the welfare and indifference functions for the two-stage global game:
\begin{multline}
W_{\text{two-stage}}(\tau, \sigma, \gamma)
\Equaldef \\
\int_{-\infty}^{1} \Big[F^{(1)}(\theta) + \gamma\big(1 - F^{(1)}(\theta)\big)\Big](1 - \theta)\,\phi(\theta)\,d\theta \\
+ \int_{1}^{\infty} F^{(1)}(\theta)\big(F^{(1)}(\theta) - \theta\big)\,\phi(\theta)\,d\theta,
\end{multline}
and
\begin{multline} \label{eq:two_stage_indifference_function}
\Delta_{\text{two-stage}}(\tau, \sigma, \gamma)
\Equaldef
\int_{-\infty}^{1} (1 - \gamma)(1 - \theta)\, f(\theta \mid Y=\tau)\,d\theta \\
+ \int_{1}^{\infty} \left(\Phi\!\left(\frac{\tau - \theta}{\sigma}\right) - \theta\right) f(\theta \mid Y=\tau)\,d\theta,
\end{multline}
where $F^{(1)}(\theta) = \Phi\!\left(\frac{\tau-\theta}{\sigma}\right)$. Similarly, for the single-stage global game,
\begin{equation}
W_{\text{single-stage}}(\tau, \sigma)
\Equaldef
\int_{-\infty}^{\infty} F^{(1)}(\theta)\big(F^{(1)}(\theta) - \theta\big)\,\phi(\theta)\,d\theta,
\end{equation}
and
\begin{multline} \label{eq:single_stage_indifference_function}
\Delta_{\text{single-stage}}(\tau, \sigma)
\Equaldef \\
\int_{-\infty}^{\infty} \left(\Phi\!\left(\frac{\tau - \theta}{\sigma}\right) - \theta\right) f(\theta \mid Y=\tau)\,d\theta.
\end{multline}

Let $\tau^\star(\sigma, \gamma)$ be the unique solution to $\Delta_{\text{two-stage}}(\tau^\star, \sigma, \gamma) = 0$, and let $\tau_{\text{single}}^\star(\sigma)$ be the unique solution to $\Delta_{\text{single-stage}}(\tau_{\text{single}}^\star, \sigma) = 0$\footnote{Uniqueness is assumed here; it follows from Theorem~\ref{theorem:unique_BNE} under $\sigma^2 < 2\pi$.}.

We first show that $\tau^\star(\sigma, 0) > \tau_{\text{single}}^\star(\sigma)$. Setting $\gamma = 0$ and using $\Delta_{\text{single-stage}}(\tau_{\text{single}}^\star, \sigma) = 0$, we compute
\begin{multline}
\Delta_{\text{two-stage}}(\tau_{\text{single}}^\star, \sigma, 0)
= \int_{-\infty}^{1} (1 - \theta)\, f(\theta \mid Y=\tau_{\text{single}}^\star)\,d\theta  \\
\quad + \int_{1}^{\infty} \left(\Phi\!\left(\frac{\tau_{\text{single}}^\star - \theta}{\sigma}\right) - \theta\right) f(\theta \mid Y=\tau_{\text{single}}^\star)\,d\theta  \\
= \int_{-\infty}^{1} \left(1 - \Phi\!\left(\frac{\tau_{\text{single}}^\star - \theta}{\sigma}\right)\right) f(\theta \mid Y=\tau_{\text{single}}^\star)\,d\theta
\;>\; 0.
\end{multline}
Since $\Delta_{\text{two-stage}}$ is strictly decreasing in $\tau$ (Theorem~\ref{theorem:unique_BNE}), this implies $\tau^\star(\sigma, 0) > \tau_{\text{single}}^\star(\sigma)$. We next show that $W_{\text{two-stage}}$ is strictly increasing in $\tau$ for $\tau < \tau^\star(\sigma, \gamma)$. Differentiating with respect to $\tau$ and using $\frac{1}{\sigma}\phi\!\left(\frac{\tau-\theta}{\sigma}\right)\phi(\theta) = f(\theta \mid Y=\tau)\,f(Y=\tau)$, we obtain
\begin{multline}
\frac{d W_{\text{two-stage}}}{d\tau}(\tau, \sigma, \gamma)
= f(Y=\tau)\,\Delta_{\text{two-stage}}(\tau, \sigma, \gamma) \\
+ f(Y=\tau)\int_{1}^{\infty} \Phi\!\left(\frac{\tau - \theta}{\sigma}\right) f(\theta \mid Y=\tau)\,d\theta.
\end{multline}
Since $\Delta_{\text{two-stage}}(\tau, \sigma, \gamma) > 0$ for $\tau < \tau^\star(\sigma, \gamma)$, and the second term is always non-negative, $W_{\text{two-stage}}$ is strictly increasing in $\tau$ on this interval. Therefore,
\begin{multline}
W_{\text{two-stage}}(\tau^\star(\sigma,0), \sigma, 0)
> W_{\text{two-stage}}(\tau_{\text{single}}^\star(\sigma), \sigma, 0)
\\ > W_{\text{single-stage}}(\tau_{\text{single}}^\star(\sigma), \sigma),
\end{multline}
where the second inequality follows because $W_{\text{two-stage}}(\tau, \sigma, 0) \ge W_{\text{single-stage}}(\tau, \sigma)$ for all $\tau$. 
 By continuity of $W_{\text{two-stage}}$ and $\tau^\star(\sigma,\gamma)$ in $\gamma$, the inequality
\begin{equation}
W_{\text{two-stage}}(\tau^\star(\sigma, \gamma), \sigma, \gamma)
>
W_{\text{single-stage}}(\tau_{\text{single}}^\star(\sigma), \sigma)
\end{equation}
persists for all $\gamma \in (0, \bar{\gamma}_\sigma)$, for some $\bar{\gamma}_\sigma > 0$.

\smallskip
\textbf{Part~II.}
As $\sigma \to 0$, the posterior $f(\theta \mid Y=\tau)$ concentrates at $\theta = \tau$. Evaluating the indifference conditions in this limit gives $\tau^\star(\sigma, \gamma) \to 1$ and $\tau_{\text{single}}^\star(\sigma) \to \frac{1}{2}$. Hence, for any fixed $\gamma \in (0,1)$, there exists $\bar{\sigma}_\gamma > 0$ such that $\tau^\star(\sigma, \gamma) > \tau_{\text{single}}^\star(\sigma)$ for all $\sigma \in (0, \bar{\sigma}_\gamma)$. The strict monotonicity of $W_{\text{two-stage}}$ established in Part~I then gives
\begin{multline}
W_{\text{two-stage}}(\tau^\star(\sigma, \gamma), \sigma, \gamma)
>
W_{\text{two-stage}}(\tau_{\text{single}}^\star(\sigma), \sigma, \gamma)
\\>
W_{\text{single-stage}}(\tau_{\text{single}}^\star(\sigma), \sigma)
\end{multline}
for all $\sigma \in (0, \bar{\sigma}_\gamma)$.\hfill\QED

\end{document}

%% file: symbols.tex
\usepackage{todonotes}

\newcommand{\ags}{[N]}

\newcommand{\indic}[1]{\mathbb{1}(#1)}